%% file: RTS_with_Reach.tex
\documentclass[a4paper,UKenglish,cleveref, autoref, thm-restate, numberwithinsect]{lipics-v2021}

\title{Regular Model Checking for Systems with Effectively Regular Reachability Relation} 

\author{Javier Esparza}{Technical University of Munich, Germany}{}{https://orcid.org/0000-0001-9862-4919}{}
\author{Valentin Krasotin}{Technical University of Munich, Germany}{}{https://orcid.org/0009-0002-2129-2754}{}

\authorrunning{J. Esparza and V. Krasotin} 

\Copyright{Javier Esparza and Valentin Krasotin} 



\relatedversion{} 

\acknowledgements{We thank Anthony Widjaja Lin and Pascal Bergsträßer for very helpful discussions.}

\nolinenumbers 

\usepackage{tabularray}

\newcommand{\prob}{\mathbb{P}}
\newcommand{\Id}{\mathsf{Id}}

\input{notation.tex}

\begin{document}

\maketitle

\begin{abstract}
Regular model checking is a well-established technique for the verification of \emph{regular transition systems} (RTS): transition systems whose initial configurations and transition relation can be effectively encoded as regular languages. In 2008, To and Libkin studied RTSs in which the reachability relation (the reflexive and transitive closure of the transition relation) is also effectively regular, and showed that the recurrent reachability problem (whether a regular set $L$ of configurations is reached infinitely often) is polynomial in the size of RTS and the transducer for the reachability relation. We extend the work of To and Libkin by studying the decidability and complexity of verifying \emph{almost-sure} reachability and recurrent reachability---that is, whether $L$ is reachable or recurrently reachable w.p.~1.  We then apply our results to the more common case in which only a regular overapproximation of the reachability relation is available. In particular, we extend recent complexity results on verifying safety using \emph{regular abstraction frameworks}---a  technique recently introduced by Czerner, the authors, and Welzel-Mohr---to liveness and almost-sure properties.
\end{abstract}

\section{Introduction}

Regular model checking is a well-established technique for the verification of infinite-state systems whose configurations can be represented as finite words over a suitable alphabet. It applies to \emph{regular transition systems} (RTS): systems whose set of initial configurations and transition relation are both regular, and presented as a finite automaton and a finite-state transducer, respectively~\cite{BouajjaniJNT00,AbdullaJNd02,AbdullaJNS04,Abdulla12,Abdulla21}. The main application of regular model checking is the verification of parameterized systems in which an arbitrarily long array or ring of finite-state processes interact~\cite{AbdullaST18}. Examples of these systems include mutual exclusion algorithms, cache coherence protocols, communication protocols, consensus algorithms, and others.


RTSs are very general; in particular, it is easy to encode Turing machines as regular transition systems, which makes all interesting analysis problems for RTSs undecidable. In~\cite{ToL08,ToL10}, To and Libkin observed that in some classes of RTSs, like pushdown systems and ground-term-rewriting systems, the reachability relation (i.e., the reflexive and transitive closure of the transition relation) is regular and one can effectively compute a transducer recognizing it. They showed that in this case the \emph{reachability} and \emph{recurrent reachability} problems become decidable and solvable in polynomial time in the size of the transducer for the reachability relation. These problems ask, given a RTS $\system$ and a regular set $L$ of configurations, whether some run of $\system$---that is, some run starting at some initial configuration of $\system$---visits $L$ at least once (reachability) or infinitely often (recurrent reachability). In the notation of CTL$^*$, they correspond to deciding the formulas \textbf{EF}\,$L$ and \textbf{EGF}\,$L$, respectively. Many other problems can be reduced to reachability or recurrent reachability

In this paper we extend the work of To and Libkin in two different directions. Both of them are motivated by our interest in applying regular model checking to the verification of liveness properties of \emph{randomized} distributed systems with an arbitrary number of processes, like many algorithms for distributed consensus or self-stabilization.

\subparagraph{Verification of almost-sure properties.} Proving liveness properties for randomized distributed systems amounts to showing that, under a class of adversarial stochastic schedulers selecting the process making the next move, something happens \emph{almost surely} (a.s.), or, in other words, that for every scheduler in the class and every initial configuration, the runs of the system starting at that configuration and satisfying a given property have probability one~\cite{LengalLMR17,LinR16}. This raises the question whether a.s.\ reachability and a.s.\ recurrent reachability are decidable when the reachability relation is effectively regular assuming, as in~\cite{ToL10},  that a transducer for the reachability relation is part of the input. We also study whether reachability and recurrent reachability holds for all possible schedulers, that is, whether $L$ is visited at least once or infinitely often by \emph{every} run starting at any initial configuration. These properties---which we call sure reachability and sure recurrent reachability and are expressed by the formulas \textbf{AF}\,$L$ and \textbf{AGF}\,$L$ in CTL$^*$---hold for a fair number of probabilistic systems, as observed in~\cite{LinR16}. We also study the complexity of some related properties, like a.s.\ termination or deadlock-freedom, which are important in applications.

Most algorithms for distributed consensus or self-stabilization are designed to work for an arbitrary but fixed number of processes, without dynamic process creation or deletion. In these systems the successors of a configuration with $k$ processes also have $k$ processes, and so both the configuration and its successor are encoded by words of length $k$. For this reason, on top of the complexity for general transducers we also investigate the complexity for the length-preserving case.

A summary of our results is shown in Tables~\ref{table:results1} and~\ref{table:results2} (recall that $\Pi_1^0$ are the co-recursively enumerable languages and $\Pi_1^1$ the corresponding level of the analytic hierarchy~\cite{odifreddi1992classical}). In all cases, the problem consists of deciding,  given an RTS, a nondeterministic transducer recognizing the reachability relation, and an NFA recognizing the set $L$ of configurations, whether the corresponding property holds. Our main results are:
\begin{itemize}
\item There is a big gap in the complexities of \textbf{EF}\,$L$ and \textbf{AF}\,$L$ (the first can be solved in polynomial time, while the second is undecidable) even in the length-preserving case.
\item A.s.\ recurrent reachability is easier to check than sure reachability in the length-preserving case. This is surprising, because, as mentioned above, sure reachability is often used as an easier-to-check approximation to a.s.\ recurrent reachability.
\item For sure properties, the general case is harder than the length-preserving case. Indeed, since they are $\Pi_1^1$- and $\Pi_1^0$-complete, respectively, there is no recursive reduction from the first to the second. This is contrary to the case of reachability and recurrent reachability, where the general case can be reduced to the length-preserving case.
\end{itemize}

\begin{table}[ht]
\begin{tblr}{l|l|c|c}
property  & CTL$^*$ & length-preserving & general\\
\hline
reachability & \textbf{EF}\,$L$   & \NL-complete & \NL-complete \\
recurrent reachability & \textbf{EGF}\,$L$ & \NL-complete (see~\cite{BergstrasserGLZ22}) & \NL-complete (see~\cite{BergstrasserGLZ22}) \\
sure reachability & \textbf{AF}\,$L$   & $\Pi_1^0$-complete & $\Pi_1^1$-complete \\
sure recurrent reachability & \textbf{AGF}\,$L$ & $\Pi_1^0$-complete & $\Pi_1^1$-complete \\
sure termination & \textbf{AF}\,$T$ & \NL-complete (see~\cite{BergstrasserGLZ22}) & \NL-complete (see~\cite{BergstrasserGLZ22}) \\
deadlock-freedom & \textbf{AG}\,$\overline T$ & \PSPACE-complete & \PSPACE-complete \\
\hline
\end{tblr}
\bigskip
\caption{Complexity of several decision problems for RTSs given a transducer for the reachability relation.
$T$ denotes the set of configurations without a successor. Note that sure termination is the special case of recurrent reachability where $L = \configurations$.}
\label{table:results1}
\end{table}

\begin{table}[ht]
\begin{tblr}{l|l|c|c}
 property & LTL & length-preserving & general\\
\hline
a.s.\ reachability & $\mathcal{P}(\textbf{F}\, L) = 1$ & $\Pi_1^0$-complete & undecidable\\
a.s.\ recurrent reachability & $\mathcal{P}(\textbf{GF}\, L) = 1$ & \PSPACE-complete & undecidable\\
a.s.\ termination & $\mathcal{P}(\textbf{F}\, T) = 1$ & \EXPSPACE-complete & undecidable\\
a.s.\ deadlock-freedom & $\mathcal{P}(\textbf{G}\, \overline T) = 1$ & \PSPACE-complete & \PSPACE-complete \\
\hline
\end{tblr}
\bigskip
\caption{Complexity of several decision problems for probabilistic RTSs given a transducer for the reachability relation.
$T$ denotes the set of configurations without a successor. Note that a.s.\ deadlock-freedom is equivalent to deadlock-freedom.}
\label{table:results2}
\end{table}

\subparagraph{Regular overapproximations of the reachability relation.} The reachability relation is regular for some classes of RTSs, but not for many others. In particular, this is not the case for most models of concurrent and distributed computing, like Vector Addition Systems (VAS, aka Petri nets) and many of their extensions~\cite{murata1989petri,BlondinR21}, lossy channel systems~\cite{AbdullaJ93,IyerN97}, broadcast protocols~\cite{EmersonN96,EsparzaFM99} or population protocols~\cite{AngluinADFP06}. However, several techniques exist for computing a regular \emph{overapproximation} of the reachability relation~\cite{AbdullaHH16,HongL24,CzernerEKW24}. In this paper, we are interested in the regular abstraction frameworks of~\cite{CzernerEKW24}. In this approach, transducers are used not only to model the transition relation of the system, but also, in the terminology of abstract interpretation, to model abstract domains~\cite{cousot2021principles}. After the user fixes an abstract domain by choosing a transducer, the system automatically computes another transducer recognizing the \emph{potential reachability} relation, the best overapproximation in the abstract domain of the reachability relation---in a certain technical sense\footnote{The precise notion of best approximation is not the same as in standard abstract interpretation.}. It is shown in~\cite{CzernerEKW24} that, while the safety problem (whether, given an RTS and a regular set of unsafe configurations, some reachable configuration is unsafe) is undecidable, the abstract safety problem (whether, given additionally a transducer for an abstract domain, some \emph{potentially} reachable configuration is unsafe) is \EXPSPACE-complete.

We study whether the decision algorithms for sure and almost-sure properties become semi-decision algorithms when one uses a regular abstraction, i.e., whether they become algorithms that always terminate and answer ``yes'' or ``don't know'' (or ``no'' and ``don't know''). We first show that this is the case for recurrent reachability and prove that  deciding whether some \emph{potential run} of the RTS visits $L$ infinitely often, is also \EXPSPACE-complete. We then study almost-sure recurrent reachability. It is easy to see that in general one does not obtain a semi-decision algorithm. However, our last result shows that one does for systems in which the set of predecessors of the set $L$ of goal configurations is effectively computable, a condition satisfied by all well-structured transition systems~\cite{AbdullaCJT96,FinkelS01} under weak conditions satisfied all the models of~\cite{murata1989petri,BlondinR21,AbdullaJ93,IyerN97,EmersonN96,EsparzaFM99,AngluinADFP06}. In particular, we prove that the abstract version of a.s.\ recurrent reachability is \EXPSPACE-complete for population protocols, whereas the problem itself is equivalent to the reachability problem for Petri nets under elementary reductions~\cite{EsparzaGLM16}, and therefore Ackermann-complete~\cite{CzerwinskiLLLM21,CzerwinskiO21,Leroux21}.

\section{Preliminaries}
\subparagraph*{Relations.} Let $R \subseteq X \times Y$ be a relation. The \emph{complement} of $R$ is the relation $\overline{R} := \{(x,y) \in X \times Y \mid (u,w) \notin R\}$. The \emph{inverse} of $R$ is the relation $R^{-1} := \{(y,x) \in Y \times X \mid (x,y) \in R\}$. The \emph{projections} of $R $ onto its first and second components are the sets $R|_1 := \{x \in X\mid \exists y \in Y \colon (x,y) \in R\}$ and $R|_2 := \{y \in Y\mid \exists x \in X \colon (x,y) \in R\}$. The \emph{join} of two relations $R \subseteq X \times Y$ and $S \subseteq Y \times Z$ is the relation $R \circ S := \{(x,z) \in X \times Z \mid \exists y \in Y \colon (x,y) \in R, (y,z) \in S\}$. The \emph{post-image} of a set $X' \subseteq X$ under a relation $R \subseteq X \times Y$, denoted $X' \circ R$ or $R(X')$, is the set $\{y \in Y \mid \exists x \in X' \colon (x,y) \in R\}$; the \emph{pre-image}, denoted $R \circ Y$ or $R^{-1}(Y)$, is defined analogously. Throughout this paper, we only consider relations where $X =\Sigma^*$ and $Y = \Gamma^*$ for some alphabets $\Sigma$, $\Gamma$. We just call them relations. A relation $R \subseteq \Sigma^* \times \Gamma^*$ is \emph{length-preserving} if $(u, w) \in R$ implies $|u|=|w|$.

\subparagraph*{Automata.} Let $\Sigma$ be an alphabet. A \emph{nondeterministic finite automaton (NFA)} over $\Sigma$ is a tuple $A = (Q,\Sigma,\delta,Q_0,F)$ where $Q$ is a finite set of \emph{states}, $\delta:Q \times \Sigma \to \mathcal{P}(Q)$ is the \emph{transition function}, $Q_0 \subseteq Q$ is the set of \emph{initial states}, and $F \subseteq Q$ is the set of \emph{final states}. A \emph{run} of $A$ on a word $w = w_1 \cdots w_l \in \Sigma^l$ is a sequence $q_0q_1 \cdots q_l$ of states where $q_0 \in Q_0$ and $\forall i \in [l]: q_i \in \delta(q_{i-1},w_i)$. A run on $w$ is \emph{accepting} if $q_l \in F$, and $A$ \emph{accepts} $w$ if there exists an accepting run of $A$ on $w$. The language \emph{recognised} by $A$, denoted $\Language{A}$, is the set of words accepted by $A$. If $|Q_0| = 1$ and $|\delta(q,a)| = 1$ for every $q\in Q, a \in \Sigma$, then $A$ is a \emph{deterministic finite automaton (DFA)}.

\subparagraph*{Convolutions and transducers.} Let $\Sigma$, $\Gamma$ be alphabets, let $\# \notin \Sigma \cup \Gamma$ be a padding symbol, and let $\Sigma_\#:= \Sigma \cup \{\#\}$ and $\Gamma_\#:= \Gamma \cup \{\#\}$. The \emph{convolution} of two words $u = a_1 \ldots a_k \in \Sigma^*$ and $w = b_1 \ldots b_l\in\Gamma^*$, denoted $\vtuple{u}{w}$,  is the word over the alphabet $\Sigma_\#  \times \Gamma_\#$ defined as follows. Intuitively, $\big[{u \atop w}\big]$ is the result of putting $u$ on top of $w$, aligned left, and padding the shorter of $u$ and $w$ with $\#$. Formally,  if $k \leq l$, then $\big[{u \atop w}\big] = \big[{a_1 \atop b_1}\big] \cdots \big[{a_k \atop b_k}\big]\big[{\# \atop b_{k+1}}\big] \cdots
\big[{\# \atop b_{l}}\big]$, and otherwise $\big[{u \atop w}\big] = \big[{a_1 \atop b_1}\big] \cdots \big[{a_l \atop b_l}\big]\big[{a_{l+1} \atop \#}\big] \cdots
\big[{a_k \atop \#}\big]$.

A \emph{transducer} over $\Sigma\times\Gamma$ is an NFA over $\Sigma_\#  \times \Gamma_\#$. The binary relation recognised by a transducer $T$ over $\Sigma\times\Gamma$, denoted $\Relation{T}$, is the set of pairs $(u, w) \in \Sigma^* \times \Gamma^*$ such that $T$ accepts $\big[{u \atop w}\big]$. 
A transducer is \emph{deterministic} if it is a DFA. A relation is \emph{regular} if it is recognised by some transducer. A transducer is \emph{length-preserving} if it recognises a length-preserving relation.

\subparagraph*{Complexity of operations on automata and transducers.} Given NFAs $A_1$, $A_2$ over $\Sigma$ with $n_1$ and $n_2$ states,  DFAs $B_1$, $B_2$ over $\Sigma$ with $m_1$ and $m_2$ states, and transducers $T_1$ over $\Sigma \times \Gamma$ and  $T_2$ over $\Gamma \times \Sigma$ with $l_1$ and $l_2$ states, the following facts are well known (see e.g.\ chapters 3 and 5 of~\cite{EB23}):
\begin{itemize}
\item there exist NFAs for $\Language{A_1} \cup \Language{A_2}$,  $\Language{A_1} \cap \Language{A_2}$, and $\overline{\Language{A_1}}$ with at most $n_1 + n_2$, $n_1 n_2$, and $2^{n_1}$ states, respectively; 
\item there exist DFAs for $\Language{B_1} \cup \Language{B_2}$,  $\Language{B_1} \cap \Language{B_2}$, and $\overline{\Language{B_1}}$ with at most $m_1 m_2$, $m_1 m_2$, and $m_1$ states, respectively;
\item there exist NFAs for $\Relation{T_1}|_1$ and $\Relation{T_1}|_2$ and a transducer for $\Relation{T_1}^{-1}$ with at most $l_1$ states;
\item there exists a transducer for $\Relation{T_1} \circ \Relation{T_2}$ with at most $l_1 l_2$ states; and
\item there exist NFAs for $\Language{A_1} \circ \Relation{T_1}$ and $\Relation{T_1} \circ \Language{A_2}$ with at most $n_1 l_1$ and $l_1 n_2$ states, respectively.
\end{itemize}

\subparagraph*{Turing machines.} We fix the definition of Turing machine used in the paper.

A \emph{Turing machine (TM)} is a tuple $M = (Q,\Sigma,\Gamma,\square,\delta,q_0,q_f)$ where
$Q$ is the set of states,
$\Sigma$ is the input alphabet,
$\Gamma \supsetneq \Sigma$ is the tape alphabet,
$\square \in \Gamma \setminus \Sigma$ is the blank symbol,
$q_0 \in Q$ is the initial state,
$q_f \in Q$ is the (only) final state, and
$\delta \colon (Q\setminus\{q_f\}) \times \Gamma \to Q \times \Gamma \times \{L,R\}$ is the transition function.
A \emph{nondeterministic Turing machine} is defined the same way as a Turing machine, with the difference that $\delta$ is now a function from $(Q\setminus\{q_f\}) \times \Gamma$ to $2^{Q \times \Gamma \times \{L,R\}}$.
A \emph{configuration} of a (nondeterministic) TM $M$ is a triple $(w_l, q, w_r)$ where $w_l$ is the tape to the left of the head of $M$, $q$ is the current state of $M$, and $w_r$ is the tape to the right of the head of $M$; the first symbol of $w_r$ is the symbol currently being read. The successor(s) of a configuration are defined as usual. A configuration $(w_l, q, w_r)$ is \emph{accepting} if $q = q_f$. A \emph{run} of a (nondeterministic) TM $M$ on an input $w$ is either a sequence $(c_i)_{i \in \N_0}$ of configurations of $M$ where $c_0 = (\varepsilon, q_0, w)$ and $c_i$ is a successor of $c_{i-1}$, or a tuple $(c_0,...,c_n)$ where $c_0 = (\varepsilon, q_0, w)$, $c_i$ is a successor of $c_{i-1}$, and $c_n$ is accepting. In the latter case, the run is called \emph{accepting}. $M$ \emph{accepts} the input $w$ if there exists an accepting run of $M$ on $w$.
\begin{remark}
Note that Turing machines are (unless specified otherwise) deterministic, have no rejecting state, and the transition function is total, i.e.\ every non-accepting configuration has a successor. In particular, a Turing machine halts iff it accepts.
\end{remark}

\subparagraph*{The arithmetical and analytical hierarchies.}
We briefly recall the definition of the first levels of the arithmetical and analytical hierarchies (see e.g.~\cite{odifreddi1992classical}, part IV). $\Sigma_1^0$ is the set of all semi-decidable problems. $\Sigma_1^1$ is the set of sets of the form $\{n \in \N \mid \exists \varphi_1,...,\varphi_k: f(n,\varphi_1,...,\varphi_n) \}$ where the $\varphi_i$ range over functions from $\N$ to $\N$ and $f$ is an arithmetic formula with arbitrary quantification over natural numbers. $\Pi_1^0$ and $\Pi_1^1$ are the sets of problems whose complement is in $\Sigma_1^0$ and $\Sigma_1^1$, respectively. $\Sigma_1^1$-hard and $\Pi_1^1$-hard problems are sometimes called \emph{highly undecidable}.

\subsection{Regular transition systems}

We recall standard notions about regular transition systems and fix some notations. Then we recall how to use regular transition systems to simulate Turing machines.

A \emph{transition system} is a pair $\system=(\configurations,\trafun)$ where $\configurations$ is the set of all possible \emph{configurations} of the system, and $\trafun \subseteq \configurations \times \configurations$ is a \emph{transition relation}. The \emph{reachability relation} $\reach$ is the reflexive and transitive closure of $\trafun$. Observe that, by our definition of post-set, $\trafun(C)$ and $\reach(C)$ are the sets of configurations reachable in one step and in arbitrarily many steps from $C$, respectively.

Regular transition systems are transition systems where configurations are represented by words and the transition relation is regular. We also define regular transition systems to contain a set of initial configurations which all runs start from.

\begin{definition}
A \emph{regular transition system (RTS)} over an alphabet $\Sigma$ is a pair $\system = (\initial, \trafun)$ where $\Sigma$ is an alphabet, the set of configurations is $\configurations = \Sigma^*$, $\initial \subseteq \configurations$ is a regular set of \emph{initial configurations} and $\trafun \subseteq \configurations \times \configurations$ is a regular transition relation. $\system$ is called \emph{length-preserving} if $\trafun$ is length-preserving. A configuration $c \in \configurations$ is called \emph{terminating} if it has no successor, i.e.\ if $\trafun(c) = \emptyset$. We denote the set of terminating configurations by $T$. A \emph{run} of $\system$ is either a sequence $(c_i)_{i \in \N_0}$ where $c_0 \in \initial$ and $\forall i \in \N: (c_{i-1},c_i) \in \trafun$ or a tuple $(c_0,...,c_n)$ where $c_0 \in \initial$, $\forall i \in [n]: (c_{i-1},c_i) \in \trafun$ and $c_n \in T$. In the latter case, the run is called \emph{terminating}.
\end{definition}

\begin{example}
\renewcommand{\t}{\bullet}
\newcommand{\n}{\circ}
\newcommand{\utuple}[2]{\big[{{#1} \atop {#2}}\big]}
We give two small examples of RTSs loosely inspired by Herman's randomized protocol for self-stabilization. We model an array of cells, each of which either holds a token ($\t$) or not ($\n$). The alphabet is $\{\langle, \t,\n, \rangle\}$ and the initial configurations are $\langle (\t + \n)^*\t (\t + \n)^* \rangle$, where $\langle$ and $\rangle$ mark the two ends of the array. A transition moves a token to a neighbour cell, ``swallowing'' its token if the cell is not empty. The language of the transducer is
\begin{center}
$\utuple{\langle}{\langle}\utuple{x}{x}^* \left(\utuple{\t}{\n} \utuple{\n}{\t} + \utuple{\n}{\t} \utuple{\t}{\n} + \utuple{\t}{\t}\utuple{\t}{\n}+\utuple{\t}{\n}\utuple{\t}{\t}\right) \utuple{x}{x}^*\utuple{\rangle}{\rangle}$
\end{center}
\noindent where $\vtuple{x}{x}$ is an abbreviation for $\vtuple{\t}{\t}+\vtuple{\n}{\n}$. The RTS is length-preserving. Intuitively, it satisfies that the set $\langle \n^* \t  \n^* \rangle$ of configurations with exactly one token is almost surely reachable and also almost surely recurrently reachable, for any probability distribution assigning non-zero probability to each transition.

Consider now another RTS in which the array can additionally grow or shrink on the right end. We can model this by adding to the language of the transducer the transitions
\begin{center}
$\utuple{\langle}{\langle}\utuple{x}{x}^*\utuple{\rangle}{\n}\utuple{\sharp}{\rangle} +  \utuple{\langle}{\langle}\utuple{x}{x}^*\utuple{\n}{\rangle}\utuple{\rangle}{\sharp}$.
\end{center}
The new RTS is no longer length-preserving, and the property above no longer holds for every probability distribution.
\end{example}



\subparagraph*{Simulation of Turing machines by regular transition systems.}
In our undecidability proofs we make use of the fact that RTSs can simulate Turing machines. More precisely, a configuration $(w_l, q, w_r)$ of a TM $M$ can be encoded as a word $w_l q w_r$. The transition function $\delta$ of $M$ can then be simulated by $\trafun$ as only the letters at positions next to the head can change. If the RTS is length-preserving, one can set $\initial := \{q_0 w\}\{\square\}^*$ where $w$ is the input to $M$ and terminate if the head of $M$ gets to the last position of the configuration of the RTS; in this case, $M$ accepts $w$ iff there exists a run of the RTS (starting in a configuration of a high enough length) which reaches a configuration containing $q_f$.

\section{Decision problems}

It was proved in~\cite{ToL08} that the problem whether an RTS with a regular reachability relation has a run visiting a set of configurations given by an NFA infinitely often is in \P\ for both length-preserving and general transition relations. Extending this result, we analyse the complexities of infinite reachability problems in an RTS with a regular reachability relation.

We start with a lemma showing that a basic problem about RTSs is $\Sigma_1^1$-complete in the general case and $\Sigma_1^0$-complete in the length-preserving case. The proof of the length-preserving case is very simple, while the general case requires to use a clever result by Harel, Pnueli, and Stavi~\cite{HarelPS83}.

\begin{lemma}\label{infiniterun}
Deciding whether an RTS $\system$ has an infinite run is $\Sigma_1^1$-complete. If $\system$ is length-preserving, the problem is $\Sigma_1^0$-complete.
\end{lemma}
\begin{proof}
In the length-preserving case, only finitely many configurations are reachable from any configuration; hence, if there exists an infinite run, there exists a path from some $c_0 \in \initial$ to some $c \in \configurations$ which visits $c$ twice. This path is a checkable certificate, proving membership in $\Sigma_1^0$. For hardness, we reduce from the problem whether a Turing machine $M$ accepts the empty input. First, we construct a TM $M'$ which behaves like $M$, but swaps acceptance and looping: $M'$ simulates $M$ while writing down all visited configurations; if $M'$ detects that $M$ visits a configuration for the second time, $M'$ accepts; if $M'$ detects that $M$ accepts, $M'$ goes in an infinite loop. We then simulate $M'$ by $\system$; $\system$ has an infinite run iff $M'$ goes in an infinite loop iff $M$ accepts the empty input.

We now show that the problem is $\Sigma_1^1$-complete in the general case by reducing from and to the following problem: Given a nondeterministic Turing machine $M$, does $M$ have an infinite run on the empty input which visits its initial state $q_0$ infinitely often? This problem is known to be $\Sigma_1^1$-complete~\cite[Proposition~5.1]{HarelPS83}.

Given $\system$, $M$ simulates a run of $\system$ by nondeterministically writing down an initial configuration, repeatedly guessing transitions in the transducer for $\trafun$ before nondeterministically stopping in a final state of the transducer, writing down the successor, and repeating the process. $M$ never visits $q_0$ during the guessing process (otherwise $M$ might never stop guessing an infinitely long successor) and every time a successor is guessed, $M$ visits $q_0$.

For the other direction, $\system$ can simulate $M$ with a counter which indicates when the next time $M$ visits $q_0$ will be. Every time $M$ visits $q_0$, $\system$ nondeterministically sets the counter to a number $k$, and every time $M$ does a transition without visiting $q_0$, $k$ is decreased by 1. If $k$ reaches 0, $\system$ terminates. (Note that $k$ can be set arbitrarily high with only one transition.)
\end{proof}

In the following problems, we assume that the input consists of an NFA for $\initial$, a transducer for $\trafun$, and, where applicable, a transducer for $\reach$ and an NFA for $L$.

\begin{theorem}\label{nonreach}
Sure reachability and sure recurrent reachability are $\Pi_1^1$-complete. If $\system$ is length-preserving, then they are $\Pi_1^0$-complete.
\end{theorem}
\begin{proof}
We prove that the complements of both problems are equivalent to the problem whether an RTS has an infinite run, both in the length-preserving and in the general case. We do that by reducing these four problems to each other in a cycle:

\begin{enumerate}[(a)]
	\item Given $\system$, $L$, and $\reach$, does there exist a run which only visits $L$ finitely often?
	\item Given $\system$, $L$, and $\reach$, does there exist a run which never visits $L$?
	\item Given $\system$, does there exist an infinite run?
	\item Given $\system$ where all configurations are initial, does there exist an infinite run?
\end{enumerate}
Here, (a) is the complement of the sure recurrent reachability problem, (b) is the complement of the sure reachability problem, and (c) is the problem from Lemma~\ref{infiniterun}.\\

(a) $\leq$ (b): A run visits $L$ finitely often if and only if it visits a configuration from which it never visits $L$ again. Hence the reduction just changes the initial configurations to $\reach(\initial)$ while leaving everything else the same.\\

(b) $\leq$ (c): Given $\system = (\initial, \trafun)$ and $L$, define $\system' := \big(\initial \setminus L, \trafun \cap \big(L \times \overline{L}\big)\big)$ and add self-loops to all terminating configurations.\\

(c) $\leq$ (d): Given $\system = (\initial, \trafun)$, we construct an RTS $\system'$ which has an infinite run starting from any configuration iff $\system$ has an infinite run. The idea is that $\system'$ checks an infinite run of $\system$ for correctness. The configurations of $\system'$ are of the form $c_0\#c_1\#\cdots\#c_n$ where $c_0,...,c_n$ are configurations of $\system$.

In the length-preserving case, $\system'$ checks (using multiple transitions) that $c_0 \in \initial$, that $(c_{i-1}, c_i) \in \trafun$ for all $i \in [n]$, and that there exists an $i<n$ such that $c_i = c_n$. $\system'$ can do that by e.g.\ working like a Turing machine and marking the already checked letters. If any checks fail, $\system'$ terminates; if all checks succeed, $\system'$ unmarks all letters and repeats the process. $\system$ has an infinite run iff $\system$ has a run which loops, i.e.\ visits the same configuration twice, and $\system'$ can check that run forever, which also results in an infinite run. For the converse, note that $\system'$ can only have an infinite run if all checks succeed.

The non-length-preserving case is similar. Instead of checking for a cycle, $\system'$ checks that the current prefix $c_0\#c_1\#\cdots\#c_n$ of a run of $\system$ is valid, and after all checks succeed, nondeterministically adds a configuration $c_{n+1}$ and restarts the process. It is clear that $\system'$ has an infinite run iff $\system$ does.\\

(d) $\leq$ (a): Given $\system = (\initial, \trafun)$ with $\initial = \configurations$, we construct an RTS $\system' = (\initial', \trafun')$ and an NFA for $L$ such that $\system'$ has a run which only visits $L$ finitely often iff $\system$ has an infinite run. We let $\system'$ behave the same way as $\system$, but add a special configuration $s$ (in the length-preserving case, we add infinitely many such configurations, one for each length), set $L = \initial' = \{s\}$ and add the transitions $(s,c)$ and $(c,s)$ for all $c \in \configurations$ to $\trafun'$. Then the reachability relation of $\system'$ is regular as all configurations (or all configurations of the same length) can reach each other. If $\system$ has an infinite run, then $\system'$ has the same infinite run which does not visit $L$; if $\system$ does not have an infinite run, then every run of $\system'$ must visit $L$ again and again.
\end{proof}

\begin{remark}
The undecidability of the sure reachability problem for arbitrary transducers can be immediately deduced from well-known results. Theorem 3.3 of~\cite{Esparza95} proves that the reachability relation of Basic Parallel Processes (BPPs)~\cite{ChristensenHM93,ChristensenHM93b} is expressible in Presburger arithmetic, and so regular under standard encodings of tuples of natural numbers as words~\cite{Haase18,esparza2023automata}. So BPPs are a special case of RTSs in which $\reach$ can be encoded as a transducer. Further, it was shown in~\cite{EsparzaK95,Esparza97} that sure reachability for BPPs is undecidable. The novelty of Theorem~\ref{nonreach} is to show that the problem remains undecidable in the length-preserving case (and to give a simpler proof for arbitrary transducers). Observe that the length-preserving case requires a novel argument: indeed, since sure reachability is $\Pi_1^1$-complete in the general case, it cannot be recursively reduced to sure reachability for length-preserving transducers, which is only $\Pi_1^0$-complete.
\end{remark}

\begin{theorem}\label{terminating}
Deadlock-freedom is \PSPACE-complete.
\end{theorem}
\begin{proof}
The set of terminating configurations is $\overline{\trafun^{-1}(\configurations)}$, and thus there exists a terminating run iff $\reach(\initial) \cap \overline{\trafun^{-1}(\configurations)} \neq \emptyset$. This is decidable in nondeterministic polynomial space by e.g.\ guessing a configuration $c$ letter by letter and checking on the fly that $c \in \reach(\initial)$ and $c \notin \trafun^{-1}(\configurations)$; indeed, this only requires enough memory to store a state of the NFA for $\reach(\initial)$ and a set of states of the NFA for $\trafun^{-1}(\configurations)$. For hardness, we can reduce from the universality problem: Given an NFA $A$, setting $\initial := \configurations$ and $\trafun := \Language{A} \times \configurations$ makes the problem equivalent to deciding whether $A$ is universal.
\end{proof}

\section{Almost sure properties}

In this section we present our results on the decidability and complexity of almost sure properties. We see the RTS $\system$ as a Markov chain, i.e.\ we have a probability measure $\prob \colon \trafun \to (0,1]$ which assigns a positive probability to every transition $(c,c') \in \trafun$ and satisfies $\forall c \in \configurations: \sum_{c' \in \trafun(c)} \prob(c,c') = 1$. This induces as usual a probability space 
$(\textit{run}(c), \mathbb{F},{\cal P})$ for every configuration $c$, where $\textit{run}(c)$ is the set of runs starting at $c$, $\mathbb{F}$ is the $\sigma$-field generated by all basic cylinders $\textit{run}(w)$ where $w$ is a finite path starting at $c$, and ${\cal P} \colon \mathbb{F} \rightarrow [0,1]$ is the unique probability function such that ${\cal P}(\textit{run}(w)) = \Pi_{i=1}^m \prob(c_{i-1}, c_i)$ where $w = (c_0,...,c_m)$.

In the length-preserving case, a.s.\ reachability, recurrent reachability and termination only depend on the topology of the Markov chain $\system$, and not on the numerical values of the probabilities of its transitions.
This fact is well-known for finite Markov chains. To show that it also holds for $\system$, which can be infinite, observe that, since the number of configurations of a given length is finite, every configuration can only reach finitely many configurations, and so $\system$ is the disjoint union of a finite or countably infinite family of finite Markov chains. It follows that our almost-sure properties hold if{}f they hold for every Markov chain in the family, and therefore do not depend on the numerical values either.

In the general case, the property depends on the numerical values. For example, consider a random walk on $\{a\}^*$ with transitions $\trafun = \{(a^n, a^{n+1}), (a^{n+1},a^n) \mid n \in \N_0\}$ and assume that $p := \prob(a^n, a^{n+1})$ is the same for all $n$. The probability that a run starting at $a$ eventually reaches $\varepsilon$ is 1 if and only if $p \leq 0.5$.

Again, we assume that the input consists of an NFA for $\initial$, a transducer for $\trafun$, and, where applicable, a transducer for $\reach$ and an NFA for $L$.

\subsection{The length-preserving case}

We make use of the fact that a run of $\system$ almost surely eventually visits a bottom strongly connected component (bSCC) of $\system$. Further, once in a bSCC, the run either terminates (if the bSCC is trivial, i.e.\ consists of a terminating configuration) or stays in that bSCC forever, visiting each configuration of the bSCC infinitely often.

\subparagraph{Almost sure reachability.}

\begin{proposition}\label{asreach}
Almost sure reachability is $\Pi_1^0$-complete for length-preserving RTSs.
\end{proposition}
\begin{proof}
If there exists an initial configuration $c_0 \in \initial$ from which reaching $L$ has probability less than 1, then there exists a bSCC reachable from $c_0$ such that neither the bSCC nor the path from $c_0$ to the bSCC intersects $L$. In other words, there exists a path $(c_0,...,c_n)$ such that $c_0 \in \initial$, $c_i \notin L$ for all $i$, and $c_n$ is in a bSCC which does not intersect $L$, i.e.\ $\reach(c_n) \subseteq \reach^{-1}(c_n)$ and $\reach(c_n) \cap L = \emptyset$. Given such a path, these conditions can be checked, proving semi-decidability of the complement and thus proving that almost sure reachability is in $\Pi_1^0$.

For hardness, we reduce from the non-emptiness problem for Turing machines. Given a TM $M$, we define $\initial$ to be the set of input configurations of $M$ with any number of blank symbols, that is, the set of all words $wv$ where $w$ is the encoding of an input configuration of $M$, and $v \in \{\square\}^*$ is a word of blanks. Further, we
let $\trafun$ simulate the transitions of $M$. Moreover, for every length, we add two special configurations $s, t$ to the RTS and add the following transitions to $\trafun$: $(c,s)$ and $(s,c)$ for every configuration $c$ of $M$, $(a,t)$ for every accepting configuration $a$ of $M$, and $(s,t)$. Then $\reach = \big((C \cup \{s\}) \times (C \cup \{s,t\})\big) \cup \{(t,t)\}$ where $C$ is the set of configurations of $M$, so $\reach$ is regular. Let $L = \{s\}$ (of all lengths). If $M$ has an input $w$ which it accepts, then there exists a run of $\system$ starting from $q_0w$ (with enough blank symbols) which reaches an accepting configuration of $M$ and then goes to $t$, never visiting $s$. This run has positive probability as the number of steps until $t$ is reached is finite, and thus the probability of visiting $L$ is not 1. Conversely, if $M$ does not have an input which it accepts, then no run can reach $t$ without visiting $s$ first, and a transition to $s$ will occur eventually almost surely, i.e.\ the probability of visiting $L$ is 1.
\end{proof}

For almost sure recurrent reachability, we introduce the following characterisation, which will also be useful in later sections.

\begin{lemma}\label{asinfinitereachlemma}
$\system$ reaches $L$ infinitely often almost surely iff $\reach(\initial) \subseteq \overline T \cap \reach^{-1}(L)$.
\end{lemma}
\begin{proof}
Every run of $\system$ eventually visits a bSCC almost surely, and thus every infinite run visits $L$ infinitely often iff every reachable bSCC is non-trivial and contains a configuration in $L$. This is the case iff no terminating configuration is reachable and $L$ can be reached from every reachable configuration, or formally, $\reach(\initial) \subseteq \overline T \cap \reach^{-1}(L)$.
\end{proof}

\begin{proposition}\label{asinfinitereach}
Almost sure recurrent reachability is \PSPACE-complete for length-preserving RTSs.
\end{proposition}
\begin{proof}
By Lemma~\ref{asinfinitereachlemma}, it suffices to show that deciding whether $\reach(\initial) \subseteq \overline T \cap \reach^{-1}(L)$ is \PSPACE-complete. One can decide this in nondeterministic polynomial space by e.g.\ guessing a configuration $c$ letter by letter and checking on the fly that $c \in \reach(\initial)$, $c \in \overline T$, and $c \notin \reach^{-1}(L)$. For that, one can run $c$ on the corresponding NFAs by memorizing the current set of states after each letter of $c$, and checking whether that set contains a final state of the NFA at the end. Since all involved NFAs have polynomial size (note that $\overline T = \trafun^{-1}(\configurations)$), this algorithm takes polynomial space, and membership in \PSPACE\ follows from $\NPSPACE=\PSPACE$.

For hardness, we can reduce from the NFA subset problem, i.e., whether $\Language{A} \subseteq \Language{B}$ for given NFAs $A, B$: By setting $\initial := \Language{A}$, $\trafun := \{(c,c) \mid c \in \configurations\}$, and $L := \Language{B}$, almost sure recurrent reachability holds iff $\Language{A} \subseteq \Language{B}$.
\end{proof}

We now consider almost sure termination, the property for which the proof is most involved. Again, we first prove a characterisation. We say that $\system$ \emph{terminates almost surely} if for every initial configuration $c$, the set of runs starting at $c$ that end in a terminating configuration has probability 1.

\begin{lemma}\label{asterminationlemma}
$\system$ terminates almost surely iff $\reach(\initial) \setminus \reach^{-1}(T) = \emptyset$.
\end{lemma}
\begin{proof}
$\reach(\initial) \setminus \reach^{-1}(T)$ is the set of all reachable configurations $c$ such that no terminating configuration can be reached from $c$. If such a $c$ exists, then reaching it has a positive probability, and since the run cannot terminate from $c$, $\system$ terminates with probability less than 1. For the converse, if terminating configurations can be reached from any reachable configuration, then no reachable non-trivial bSCCs exist, i.e.\ every reachable bSCC is a terminating configuration. Since a run of $\system$ eventually reaches a bSCC almost surely, the statement follows.
\end{proof}

\begin{proposition}\label{astermination}
Almost sure termination is \EXPSPACE-complete for length-preserving RTSs. If an NFA for the set of terminating configurations is provided in the input, almost sure termination becomes \PSPACE-complete.
\end{proposition}
\begin{proof}
From the characterisation of Lemma~\ref{asterminationlemma}, it is easy to show that the problem is \PSPACE-complete in the case where an NFA for $T$ is given in the input, one can e.g.\ reduce from and to the NFA subset problem. For membership in \EXPSPACE, observe that $T = \overline{\trafun^{-1}(\configurations)}$ has an NFA of exponential size in the input. One can thus build a polynomial NFA for $\reach(\initial)$ and an exponential NFA for $\reach^{-1}(T)$, guess a configuration $c$ letter by letter and check that $c \in \reach(\initial)$ and $c \notin \reach^{-1}(T)$. This is a nondeterministic algorithm using exponential space, and membership follows from $\mathsf{NEXPSPACE}=\EXPSPACE$.

For \EXPSPACE-hardness, we reduce from the problem whether an exponentially space-bounded Turing machine $M = (Q,\Sigma',\Gamma,\delta,\square,q_0,q_f)$ accepts the empty input. The configurations of our RTS $\system$ are going to be encodings of runs of $M$, and a run of $\system$ cannot terminate if and only if it starts with the encoding of the accepting run of $M$ on the empty input. We encode a run of $M$ as a word consisting of encodings of configurations of $M$ separated by $\#$.

Let $n:=|M|$ and let $p_1$, ..., $p_n$ be the first $n$ prime numbers. Then $\sum_{i=1}^n p_i$ is polynomial in $n$ by the prime number theorem while $m := \prod_{i=1}^n p_i$ is exponential in $n$; we can thus assume that the run of $M$ on the empty input uses only the first $m$ tape cells. An NFA recognizing the language of words not divisible by $p_i$ only needs $p_i$ states. Thus, by putting $n$ NFAs side by side, each of which only accepts words of length not divisible by $p_i$, one can construct an NFA of polynomial size in $n$ which only accepts words of length not divisible by $m$. This fact will be used in our reduction.

Let $\Sigma_1 := Q \cup \Gamma \cup \{\#\}$ and $\Sigma_2 := \{x' \mid x \in \Sigma_1\}$. We call symbols in $\Sigma_1$ \emph{unmarked} and symbols in $\Sigma_2$ \emph{marked}. We denote the symbol in position $i$ of a configuration of $\system$ by $x_i$. Let $\Sigma := \Sigma_1 \cup \Sigma_2$, $\configurations := \Sigma^*$ and $\initial := \{\# q_0\} \{\square\}^* \{\#\} \Sigma_1^* \{q_f\}$, i.e.\ the first configuration of $M$ in an initial configuration of $\system$ is the empty input to $M$, and the last symbol is the accepting state of $M$. $\trafun$ has three types of transitions: $\trafun_\mathit{mark}$, which marks some symbols, and $\trafun_\mathit{unmark}$, which unmarks all symbols, and $\trafun_\mathit{end}$, which is used to terminate certain runs. Specifically, $\trafun_\mathit{mark}$ nondeterministically chooses a position $i$ and a distance $k>2$ and marks the symbols at positions $i-1$, $i$, $i+1$, $i+2$, and $i+k$. For this transition to be executed, three conditions are required:
\begin{itemize}
	\item No symbol is marked.
	\item There is exactly one $\#$ in the subword $x_{i+1} \cdots x_{i+k}$.
	\item If $x_{i-1}x_ix_{i+1}x_{i+2}$ is part of a configuration of $M$, then $x_k$ is \emph{not} the symbol appearing in place of $x_i$ in the next configuration of $M$. In particular, if $x_i = \#$, then $x_k \neq \#$.
\end{itemize}
The last point can be intuitively described as the encoding of the run having a ``mistake'' at position $x_k$. This description only fits if $\trafun_\mathit{mark}$ chooses the correct distance; for other distances, even the correct run has a mistake. It will be relevant, however, that the encoding of the correct run does not have a mistake if distance $m$ is chosen.

$\trafun_\mathit{unmark}$ takes as input a configuration with exactly five marked symbols, checks that the distance between the second and the fifth is not a multiple of $m$, and unmarks all symbols. If the distance is a multiple of $m$, the transition cannot occur. As explained above, this can be done with a transducer of polynomial size in $n$.

$\trafun_\mathit{end}$ takes as input a configuration where the distance between any two $\#$s is not a multiple of $m$, and replaces all symbols by $\#'$, after which the run terminates.

Observe that transitions only mark and unmark symbols and never change them. Thus every reachable configuration can be reached in at most two steps, and $\reach = \{(c,c) \mid c \in \configurations\} \cup \trafun \cup \trafun^2$ is regular and has a transducer of polynomial size in $n$. Furthermore, a run of $\system$ can only terminate if either the initial configuration has $\#$s at a distance not divisible by $m$, or the run reaches a configuration where the distance between the marked symbols $x_i$ and $x_{i+k}$ chosen by $\trafun_\mathit{mark}$ is a multiple of $m$, i.e.\ $\trafun_\mathit{mark}$ finds a mistake at a distance divisible by $m$.

If $M$ accepts the empty input, a run of $\system$ can start with the encoding of the accepting run of $M$. Since the run only uses the first $m$ tape cells, there exists such an encoding where the distance between two consecutive $\#$ is always $m$; therefore, the run of $\system$ can never terminate with $\trafun_\mathit{end}$. Moreover, $\trafun_\mathit{mark}$ can only choose distances $k \leq 2m-1$ as otherwise there will be two $\#$ between position $i+1$ and $i+k$. If $k \neq m$, then $\trafun_\mathit{unmark}$ will unmark the positions in the next transition; and $k = m$ cannot be chosen because the correct run has no ``mistakes'' for $\trafun_\mathit{mark}$ to find. Hence such a run cannot terminate.

If $M$ does not accept the empty input, then there exists no run of $M$ starting with the empty input and ending with an accepting configuration. That is, every encoding of such a run in $\system$ must have a mistake, i.e.\ there must exist symbols $x_{i-1}$, $x_i$, $x_{i+1}$, $x_{i+2}$, $x_{i+lm}$ where $x_{i+lm}$ does not result from $x_{i-1}x_ix_{i+1}x_{i+2}$, where $l \in \N$ s.t.\ $lm$ is the distance between two consecutive $\#$. Such a mistake will almost surely be eventually ``found'' by $\trafun_\mathit{mark}$, and the run will terminate. If the distance between two consecutive $\#$ is not a multiple of $m$, then the run can terminate with $\trafun_\mathit{end}$. In both cases, the run will almost surely eventually terminate.\\
\end{proof}

\subsection{The general case}

In the general case, not every run ends up in a bSCC almost surely; bSCCs do not even necessarily exist. Furthermore, the probabilities of the transitions can affect the answer, which introduces the question of how $\prob$ is provided in the input, and how to deal with configurations with infinitely many successors. We show that the problems are undecidable even in the special case where every configuration has finitely many successors and $\prob$ is the uniform probability measure, i.e.\ $\forall (c,c') \in \trafun: \prob(c,c') = \frac{1}{|\trafun(c)|}$.

We prove undecidability of the almost sure termination problem; the other problems can be easily reduced to it. For that, we reduce from the problem whether a deterministic Turing machine $M$ loops on the empty input, i.e.\ reaches the same configuration twice. A run of our RTS $\system$ will simulate the run of $M$ on the empty input with high probability and have a low probability (the longer the configuration of $\system$, the lower) to jump to a special configuration $s$. From $s$, the run can terminate, restart the simulation of the run of $M$ on the empty input, or, with low probability, jump to any configuration of $M$ (this is done to ensure that $\reach$ is regular). If $M$ halts or the head goes too far to the right, $\system$ restarts the simulation of $M$ on the empty input while increasing the length by 1. This way, if $M$ loops on the empty input, the run will eventually jump to $s$ and have a chance to terminate. If $M$ does not loop on the empty input, but instead halts or the head goes arbitrarily far, then, with high probability, the configuration of $\system$ will become longer and longer, and the probability of reaching $s$ will decrease, leading to termination with probability $<1$. For the full proof, see Appendix~\ref{app:asterminationnlp}.

\begin{restatable}{proposition}{asterminationnlp}\label{asterminationnlp}
Almost sure termination is undecidable.
\end{restatable}

\begin{corollary}
Almost sure reachability and almost sure recurrent reachability are undecidable.
\end{corollary}
\begin{proof}
One can reduce from almost sure termination by adding a new configuration $a$, setting $L := \{a\}$, and adding transitions $(a,a)$ and $(t,a)$ for every $t \in T$.
\end{proof}

\section{Sure and almost-sure properties in regular abstraction frameworks}
As mentioned in the introduction, while some formalisms that can be modeled as RTSs have an effectively computable regular reachability relation, many others do not, including lossy channel systems, Vector Addition Systems and their many extensions, broadcast protocols, and population protocols. On the other hand, there exist different techniques to effectively compute a regular \emph{overapproximation} of $\reach$, which in this section we generically call \emph{potential reachability} and denote $\potreach$.
We consider the particular case of \emph{regular abstraction frameworks}, a recent technique for the computation of $\potreach$~\cite{CzernerEKW24} introduced by the authors and colleagues. The technique automatically computes an overapproximation $\potreach$  for any \emph{regular abstraction} representable---in a certain technical sense--- by a transducer, and the computational complexity of the automatic procedure is precisely known (see Section~\ref{subsec:raf} below).

The section is structured as follows. Section~\ref{subsec:raf} recalls the main notions and results of~\cite{CzernerEKW24}. Sections~\ref{subsec:rafsure} and Sections~\ref{subsec:rafalmostsure} study sure and almost-sure properties, respectively.

\subsection{Regular abstraction frameworks}
\label{subsec:raf}
We briefly recall the main idea behind regular abstraction frameworks. Recall that the configurations of an RTS $\system$ are finite words over an alphabet $\Sigma$. A regular abstraction framework introduces a second alphabet $\Gamma$ and a deterministic transducer $\interpretation$ over $\Gamma \times \Sigma$, called an \emph{interpretation}. Words over $\Gamma$ are called \emph{constraints}. Intuitively, words over $\Gamma$ are ``identifiers'' for sets of configurations, and $\interpretation$ ``interprets'' an identifier $w \in \Gamma^*$ as the set of configurations $\interpretation(w):= \{w\} \circ \Relation{\interpretation}$.  A configuration $c$ satisfies a constraint $w$ if $c \in \interpretation(w)$.

A constraint $w \in \Gamma^*$ is \emph{inductive} if $\interpretation(w) \circ \trafun \subseteq \interpretation(w)$. Given an inductive constraint $w \in \Gamma^*$, we have $\interpretation(w) \circ \reach \subseteq  \interpretation(w)$, i.e., $\interpretation(w)$ is closed under the reachability relation. An inductive constraint $w$ \emph{separates} two configurations $c, c'$ if $c \in \interpretation(w)$ and $c' \notin \interpretation(w)$. Observe that if $w$ separates $(c, c')$, then $c'$ is not reachable from $c$. This leads to the definition of the \emph{potential reachability relation induced by $\interpretation$}, denoted $\potreach$, as the set of all pairs $(c, c')$ that are \emph{not} separated by any inductive constraint. Formally:
\[ \potreach = \{ (c, c') \in \configurations \times \configurations \mid \forall w \in \Gamma^* \colon \mbox{$w$ is not inductive or does not separate $(c, c')$} \} \]

The following theorem was shown in~\cite{CzernerEKW24}.

\begin{theorem}
\label{thm:abs}
Let $\system=(\initial,\trafun)$ be an RTS and let $\interpretation$ be an interpretation.
\begin{enumerate}
\item $\potreach$ is reflexive, transitive, and regular. Further, one can compute in exponential time a transducer recognizing $\overline{\potreach}$ with at most $n_\interpretation^2 \cdot 2^{n_\Delta \cdot n_\interpretation^2}$ states.
\item Given an NFA for $L$, deciding whether some configuration in $L$ is potentially reachable from some initial configuration (formally: whether $\potreach(\initial) \cap L = \emptyset$ holds) is \EXPSPACE-complete.
\end{enumerate}
\end{theorem}

The problem of Theorem~\ref{thm:abs}.2 is called \textsc{Abstract Safety} in~\cite{CzernerEKW24}. We state a corollary.

\begin{proposition}
Abstract deadlock-freedom is \EXPSPACE-complete.
\end{proposition}
\begin{proof}
Abstract deadlock-freedom is the problem whether $\potreach(\initial) \cap T = \emptyset$. A nondeterministic algorithm can guess a configuration $c \in \configurations$ and check that $c \notin \overline{\potreach(\initial)}$ and $c \notin \overline{T}$. Since one can construct an exponential-sized NFA for $\overline{\potreach(\initial)}$ and a polynomial-sized NFA for $\overline{T}$, this algorithm needs exponential space.

For hardness, we reduce from \textsc{Abstract Safety}. Given an RTS $\system$ and a set $L$ of unsafe configurations, we construct $\system'$ by adding a new configuration $t$ (in the length-preserving case, adding such a configuration for each length), adding self-loops to all configurations except $t$, adding transitions $(c,t)$ for every $c \in L$, and adding $t$ to all constraints (i.e.\ for every constraint $\varphi$, we have $(\varphi, t) \in \interpretation$). This makes $t$ the only terminating configuration while not affecting the inductivity of any constraints. It is easy to see that abstract safety holds for $\system$ iff abstract deadlock-freedom holds for $\system'$.
\end{proof}

\subsection{Sure properties}
\label{subsec:rafsure}
We define \textsc{Abstract Liveness}, obtained by substituting  ``recurrent reachability'' for ``reachability'' in \textsc{Abstract Safety}.

\begin{definition}
\label{def:potrun}
Let $\system=(\initial,\trafun)$ be an RTS and let $\interpretation$ be an interpretation. A \emph{potential run} of $\system$ is a run of the RTS $(\initial,\trafun')$ with $\trafun' := (\potreach \setminus \Id) \cup \trafun$. The \emph{abstract liveness} problem is defined as follows:
\begin{quote}
\textsc{Abstract Liveness}\\
Given: RTS $\system=(\initial,\trafun)$, interpretation $\interpretation$, NFA $A$ for $L$. \\
Decide: Does some potential run of $\system$ visit $L$ infinitely often?
\end{quote}
\end{definition}

\begin{restatable}{lemma}{closures}\label{lem:closures}
The reflexive transitive closure of $\trafun'$ is $\potreach$. The transitive closure of $\trafun'$ is $\trafun' \cup (\potreach \setminus \Id)^2$.
\end{restatable}
\begin{proof}
See Appendix~\ref{app:closures}.
\end{proof}

The first statement of Lemma~\ref{lem:closures} shows that our definition of a potential run coincides with the definition of abstract safety in~\cite{CzernerEKW24}: Abstract safety holds iff no potential run reaches $L$. The second statement will be used to characterise abstract liveness.

In the rest of the section, we show that \textsc{Abstract Liveness} and abstract sure termination, which can be seen as a special case of \textsc{Abstract Liveness}, is \EXPSPACE-complete. We start with a lemma, similar to the starting point of To and Libkin in~\cite{ToL08}.

\begin{lemma}
\label{lem:aandb}
Let $\system=(\initial,\trafun)$, $\interpretation$, $A$ be an instance of \textsc{Abstract Liveness}. Let $\trafun'$ be as in Definition~\ref{def:potrun}. There exists a potential run of $\system$ that visits $\Language{A}$ infinitely often if and only if one of these conditions hold:
\begin{enumerate}[(a)]
	\item There exist configurations $c_0$, $c$ such that $c_0 \in \initial$, $c \in L$, $(c_0,c) \in \potreach$, and $(c,c) \in \trafun \cup (\potreach \setminus \Id)^2$.
	\item There exists a sequence $(c_i)_{i \in \N_0}$ of pairwise distinct configurations such that $c_0 \in \initial$, $\forall i \in \N: c_i \in L$, and $\forall i \in \N_0: (c_i,c_{i+1}) \in \potreach$.
	\end{enumerate}
\end{lemma}
\begin{proof}
We use Lemma~\ref{lem:closures}. If (a) holds, then the infinite sequence $c_0, c, c, \cdots$ is a potential run. If (b) holds, then $(c_i)_{i \in \N_0}$ is a potential run. For the converse, let $(c_i)_{i \in \N_0}$ be a potential run that visits $L$ infinitely often. If $c_i = c_j \in L$ for some $i < j$, then $(c_i,c_i) \in \trafun' \cup (\potreach \setminus \Id)^2$, and since $(c_i,c_i) \in \Id$, (a) holds. Otherwise, removing all configurations not in $L$ from the sequence yields a sequence of pairwise distinct configurations in $L$, and (b) holds.
\end{proof}

The sequence in condition (b) is called an \emph{infinite directed clique}, see~\cite{BergstrasserGLZ22}. Note that (b) never holds for length-preserving RTSs. We prove that both (a) and (b) can be decided in \EXPSPACE.

\begin{lemma}\label{abstractinfinitereach}
Deciding (a) is in \EXPSPACE.
\end{lemma}
\begin{proof}
One can guess a configuration $c \in \configurations$ letter by letter and check that $c \in L$, $c \notin \overline{\potreach(\initial)}$ and $(c,c) \in \trafun \cup (\potreach \setminus \Id)^2$. $(c,c) \in (\potreach \setminus \Id)^2$ is equivalent to there existing a $c' \in \configurations$ such that $c' \neq c$ and $(c,c'), (c',c) \notin \overline \potreach$. Since an NFA for $\overline{\potreach(\initial)}$ can be built using exponential space (see Theorem~\ref{thm:abs}), the algorithm amounts to guessing $c, c'$ letter by letter (in parallel) and running words on NFAs of at most exponential size. To do this, one can memorize the current set of states after each letter and check whether that set contains a final state of the NFA at the end. Since all involved NFAs are at most exponential, this algorithm takes exponential space, and the result follows from $\mathsf{NEXPSPACE}=\EXPSPACE$.
\end{proof}

For (b), we first need the following lemma, based on Lemma~4 from~\cite{ToL08}.
\begin{lemma}\label{alphabeta}
Condition (b) holds if and only if there exist sequences $(\alpha_i)_{i \in \N_0}$ and $(\beta_i)_{i \in \N_0}$ in $\configurations$ such that
\begin{enumerate}
    \item $\alpha_0 \in \initial$ and $|\alpha_i| > 0$ for every $i \in \N$;
    \item $|\alpha_i| = |\beta_i|$ for every $i \in \N_0$;
    \item there exists an infinite run $r$ of $A$ on $\beta_0\beta_1\cdots$ such that, for every $i \in \N_0$, $\alpha_{i+1}$ is accepted from the state reached by $r$ on the prefix $\beta_0\cdots\beta_i$,
    \item there exists an infinite run $r'$ of the transducer for $\potreach$ on $(\beta_0\beta_1\cdots, \beta_0\beta_1\cdots)$ such that $\forall i \in \N_0: (\alpha_i, \beta_i\alpha_{i+1})$ is accepted from the state reached by $r'$ on the prefix $(\beta_0\cdots\beta_{i-1}, \beta_0\cdots\beta_{i-1})$.
\end{enumerate}
\end{lemma}
\begin{proof}
One direction is easy: if the four conditions hold, the sequence $(c_i)_{i \in \N_0}$ with $c_i := \beta_0 \cdots \beta_{i-1} \alpha_i$ satisfies (b). For the other direction, we make use of the transitivity of $\potreach$, see the proof of Lemma~4 in~\cite{ToL08}.
\end{proof}

\begin{restatable}{lemma}{abstractinfinitereachnlp}\label{abstractinfinitereachnlp}
Deciding (b) is in \EXPSPACE.
\end{restatable}
\begin{proof}
See Appendix~\ref{app:abstractinfinitereachnlp}.
\end{proof}

\begin{theorem}\label{abslivenessEXPSPACE}
\textsc{Abstract Liveness} is \EXPSPACE-complete both for length-preserving and for general RTSs.
\end{theorem}
\begin{proof}
Membership in \EXPSPACE\ follows from Lemmas~\ref{lem:aandb},~\ref{abstractinfinitereach}, and~\ref{abstractinfinitereachnlp}. \EXPSPACE-hardness follows from an easy reduction from the abstract safety problem: One can add self-loops to all configurations in $\Language{A}$ to make sure that (a) holds iff $\potreach(\initial) \cap \Language{A} \neq \emptyset$. Since self-loops do not affect whether constraints are inductive, $\potreach$ does not change. Note that for length-preserving RTSs, abstract infinite reachability is equivalent to (a).
\end{proof}

\begin{restatable}{theorem}{abssuretermination}\label{thm:abssuretermination}
Abstract sure termination is \EXPSPACE-complete both for length-preserving and for general RTSs.
\end{restatable}
\begin{proof}
Abstract sure termination is the special case of abstract liveness where $L = \configurations$ and therefore also in \EXPSPACE. The hardness proof is more involved, see Appendix~\ref{app:abssuretermination}.
\end{proof}

\subsection{Almost-sure properties}
\label{subsec:rafalmostsure}
Unlike recurrent reachability, almost-sure recurrent reachability cannot be semi-decided using only an overapproximation of the reachability relation, not even in the length-preserving case. To see why, recall that a run of a length-preserving RTS  visits a regular set $L$ of configurations infinitely often with probability 1 iff $\reach(\initial) \subseteq \overline T \cap \reach^{-1}(L)$ (Lemma~\ref{asinfinitereachlemma}). Unfortunately, given an overapproximation $\potreach \supseteq \reach$, this condition neither implies nor is implied by $\potreach(\initial) \subseteq \overline T \cap \potreach^{-1}(L)$.

This situation changes when $\reach^{-1}(L)$ is an effectively computable regular set of configurations.
Indeed, in this case $\potreach(\initial) \subseteq \overline T \cap \reach^{-1}(L)$ is decidable, and implies $\reach(\initial) \subseteq \overline T \cap \reach^{-1}(L)$. In the rest of this section, we show that this observation leads to positive results for:
\begin{quote}
\textsc{Abstract A.S. Liveness}\\
Given: RTS $\system=(\initial,\trafun)$, interpretation $\interpretation$, NFA $A$ for $L$. \\
Decide: Does a potential run of $\system$ visit $L$ infinitely often almost surely?
\end{quote}


Broadcast protocols~\cite{EmersonN96,EsparzaFM99} and population protocols~\cite{AngluinADFP06} are two models of distributed computation in which an arbitrarily large but fixed number of identical finite-state processes interact. In an action of a broadcast protocol, a process broadcasts a message to all other processes, which change their state according to a transition function. In an action of a population protocol, only two agents interact and change state.
A configuration of a broadcast or population protocol with $n$ processes can be modeled as a word over the finite set of states $Q$ of the processes of length $n$. It is easy to see that the transition functions of both formalisms are regular, and, since the number of processes remains constant throughout the computation, they are instances of length-preserving RTSs.

The \emph{scattered subword} order on configurations is defined by: for every $c, c' \in Q^*$, we have $c \preceq c'$ if $c=q_1 \cdots q_n \in Q^*$ and $c'=w_0q_1w_1q_2 \cdots q_n w_n$ for some words $w_0, \cdots, w_n \in Q^*$. A set $L$ of configurations is upward-closed w.r.t.\ $\preceq$ if $c \in L$ and $c \preceq c'$ implies $c' \in L$. The following well-known theorem follows from theory of well-structured transition systems (WSTS), and the fact that population protocols can be seen as a special case of probabilistic VAS:

\begin{theorem}{(\cite{FinkelS01}, Theorem 3.6, \cite{Rackoff76})}
Let $\system$ be an RTS modeling a broadcast protocol or a population protocol, and let $L$ be an upward-closed set of configurations. Then $\reach^{-1}(L)$ is regular and effectively computable. Further, in the case of population protocols, $\reach^{-1}(L)$ can be computed in exponential space in the size of $\system$ and an NFA recognizing $L$.
\end{theorem}

Notice that this is the case even though $\reach$ is not regular for either broadcast or population protocols. As a corollary of this theorem and the observations above, we obtain:

\begin{theorem}
\textsc{Abstract A.S. Liveness} is decidable for RTS modeling broadcast or population protocols and upward-closed sets of configurations. Moreover, it is \EXPSPACE-complete for population protocols.
\end{theorem}
\begin{proof}
One needs to decide whether $\potreach(\initial) \subseteq \overline T \cap \reach^{-1}(L)$, i.e.\ whether $\overline{\potreach(\initial)} \cup (\overline T \cap \reach^{-1}(L))$ is universal. An NFA for $\overline{\potreach(\initial)}$ can be computed in exponential space (see Theorem~\ref{thm:abs}), and the statement follows.
\end{proof}

\section{Conclusions}

We have extended the work of To and Libkin on recurrent reachability of RTSs~\cite{ToL08,ToL10} to the verification of sure and almost-sure properties, and applied our results to the setting of regular abstraction frameworks~\cite{CzernerEKW24}. In particular, we have shown that \textsc{Abstract Liveness} has the same complexity as \textsc{Abstract Safety}, and that \textsc{Abstract A.S. Liveness} is decidable for some important models of distributed systems.

\bibliography{main.bib}

\appendix

\section{Proof of Proposition~\ref{asterminationnlp}}
\label{app:asterminationnlp}

\asterminationnlp*
\begin{proof}
Let $M = (Q,\Sigma',\Gamma,\delta,\square,q_0,q_f)$ be a TM. We construct an RTS $\system$ which terminates almost surely iff $M$ loops on the empty input. The configurations of $\system$ are encodings of configurations of $M$ together with a binary counter: For a configuration $c$ of $M$ of length $n$, we encode pairs $(c,k)$ for $k \in \{0,...,2^n-1\}$ as interleaved words in $((Q \cup \Gamma) \{0,1\})^n$, i.e.\ words of the form $c_1 0 c_2 1 c_3 1 \cdots c_n 0$ of length $2n$. For every length, we add a special configuration $s$ which is also interleaved with a number, and we add a configuration $t$ with no successor.

We define $\initial := \{(q_0 \varepsilon, 0)\}$ of every length (an infinite set). $\trafun$ can either increase the counter by 1, or simulate $M$ while resetting the counter to 0. If the counter reaches $2^n-1$, the run can jump to the configuration $(s,0)$ of its length. From a configuration of the form $(s,k)$, $\trafun$ can either increase the counter by 1, jump to $(q_0\varepsilon,0)$ of its length, or jump to $t$. From $(s,2^n-1)$ (i.e.\ $s$ interleaved with ones), the run can jump to any configuration $(c,0)$ with $|c| = n$. Furthermore, every halting configuration and every configuration where the head of $M$ is in the rightmost position can jump to $q_0\varepsilon$ while increasing the configuration's length by 2 (i.e.\ the configuration of $M$ and the counter gain one position each). As a result, every pair $(c',k')$ is reachable from every pair $(c,k)$ with $|c| \leq |c'|$, and $\reach$ is regular.

Observe that if a run is in a configuration $(c,0)$ with $|c| = n$, it must almost surely either reach the configuration $(s,0)$ with $|s| = n$ or reach a halting configuration of $M$ or a configuration of $M$ where the head is in the rightmost position and jump to $(q_0\varepsilon,0)$ of length $n+1$. Therefore a run almost surely either terminates in $t$ or increases the length of its configurations indefinitely, and in the latter case, must visit $(q_0\varepsilon,0)$ (of increasing lengths) infinitely often.

Assume that $M$ loops on $q_0\varepsilon$. If $\system$ starts with a configuration where the loop is not contained in the size of the configuration, the run will increase the length of the configuration until the loop is contained (or terminate, with low probability). From this point on, every time the run visits $(q_0\varepsilon,0)$, it must later visit $(s,0)$ of the same length, and every time it does so, it terminates with a probability of almost $\frac{1}{2}$. Since this will happen arbitrarily often, the run terminates almost surely.

Now assume that $M$ does not loop on $q_0\varepsilon$. Observe that from any configuration $(c,0)$ with $|c| = n$, visiting $s$ before visiting $(\delta(c),0)$ has probability $2^{-2^n}$. Since $M$ does not loop, there are at most $|Q \cup \Gamma|^n$ possible configurations of $M$ the run of $\system$ can visit before it reaches a halting configuration or a configuration where the head of $M$ is in the rightmost position. Hence, from $(q_0\varepsilon,0)$ of length $n$, the probability of visiting $(s,0)$ of length $n$ before visiting $(q_0\varepsilon,0)$ of length $n+1$ is at most
\[ |Q \cup \Gamma|^n 2^{-2^n}, \]
and the probability of reaching any $(s,0)$ at all from $(q_0\varepsilon,0)$ with $|q_0\varepsilon| = n$ is at most
\[ \sum_{k = n}^\infty |Q \cup \Gamma|^k 2^{-2^k}. \]
This sum converges, so there exists an $n \in \N_0$ such that it is less than 1. Since every terminating run visits $(s,0)$ at least once, the probability of a terminating run is also less than 1.
\end{proof}

\section{Proof of Lemma~\ref{lem:closures}}
\label{app:closures}

\closures*
\begin{proof}
Since $\potreach$ is reflexive, transitive, and a superset of $\trafun$, the first statement follows.

For the second statement, let $(c,c'') \in (\trafun')^2$. Then there exists $c' \in \configurations$ such that $(c,c'), (c',c'') \in \trafun'$. If $c = c'$ or $c' = c''$, then $(c,c'') \in \trafun'$. If $c \neq c'$, then $(c,c') \in \potreach \setminus \Id$ (since $\trafun \subseteq \potreach$); the same holds for $c'$ and $c''$. It follows that $(\trafun')^2 = \trafun' \cup (\potreach \setminus \Id)^2$. Now let $(c,c''') \in (\trafun')^3$. Then there exist $c',c'' \in \configurations$ such that $(c,c'), (c',c''), (c'',c''') \in \trafun'$. If any two consecutive configurations are equal, then $(c,c''') \in (\trafun')^2$. If no two consecutive configurations are equal, we have $(c,c'), (c',c''), (c'',c''') \in \potreach \setminus \Id$. Note that $(c,c'') \in \potreach$ by transitivity. If $c = c''$, then $(c,c''') \in \trafun'$; otherwise $(c,c'') \in \potreach \setminus \Id$ and therefore $(c,c''') \in (\potreach \setminus \Id)^2$. Hence $(\trafun')^3 \subseteq (\trafun')^2$, and, by inductivity, the transitive closure of $\trafun'$ is $\trafun' \cup (\trafun')^2 = \trafun' \cup (\potreach \setminus \Id)^2$.
\end{proof}

\section{Proof of Lemma~\ref{abstractinfinitereachnlp}}
\label{app:abstractinfinitereachnlp}

\abstractinfinitereachnlp*
\begin{proof}
We check whether Lemma~\ref{alphabeta} holds. For that, we first use projection and product construction to construct a transducer $T = (Q_T,\Sigma^2,\delta_T,Q_{0T},F_T)$ for $\potreach \cap (\Sigma^* \times L)$. ($T$ is double-exponential, and we will explain later how we get rid of one exponent.) Since $T$ is the product automaton of the transducer for $\potreach$ and the transducer for $\Sigma^* \times A$, which is isomorphic to $A$, the infinite runs from points 3 and 4 of Lemma~\ref{alphabeta} can be combined into one. That is, points 3 and 4 can now be written as the following condition:
\begin{gather*} \text{There exists an infinite run } r \text{ of } T \text{ on } (\beta_0\beta_1\cdots, \beta_0\beta_1\cdots) \text{ such that}\\
\forall i \in \N: (\alpha_i, \beta_i\alpha_{i+1}) \text{ and } (\#,\alpha_i) \text{ are accepted from the state reached by } r\\
\text{on the prefix } (\beta_0\cdots\beta_{i-1}, \beta_0\cdots\beta_{i-1}), \text{ and } (\alpha_0,\beta_0\alpha_1) \in \Relation{T}). \end{gather*}

Let $G$ be a directed graph with the set $(Q_T)^2$ of nodes such that there is an edge from $(q_1,q_2)$ to $(q'_1,q'_2)$ iff there exist $c,d \in \Sigma^*$ with $|c| = |d| > 0$ and $q_f \in F_T$ such that
\[ q_1 \xrightarrow{\big[{\# \atop c}\big]} q_f, \text{ } q_2 \xrightarrow{\big[{c \atop d}\big]} q'_1, \text{ } q_2 \xrightarrow{\big[{d \atop d}\big]} q'_2 \text{ in } T. \]
We show that the conditions of Lemma~\ref{alphabeta} hold iff there exists a cycle in $G$ reachable from a node $(q_{1,0},q_{2,0})$ such that there exist $\alpha_0$, $\beta_0$ with $\alpha_0 \in \initial$, $|\beta_0| = |\alpha_0|$, $q_{1,0} = \delta_T(q_0, (\alpha_0,\beta_0))$ and $q_{2,0} = \delta_T(q_0, (\beta_0,\beta_0))$.

Assume that there exist sequences $(\alpha_i)_{i \in \N_0}$, $(\beta_i)_{i \in \N_0}$ in $\Sigma^*$ such that the conditions of Lemma~\ref{alphabeta} hold. We define the sequence $(q_{1,i},q_{2,i})_{i \in \N_0}$ of nodes of $G$ where $q_{1,i}$ is the state of $T$ reached by $r$ on $(\beta_0\cdots\beta_{i-1}\alpha_i, \beta_0\cdots\beta_{i-1}\beta_i)$ and $q_{2,i}$ the state of $T$ reached by $r$ on $(\beta_0\cdots\beta_i, \beta_0\cdots\beta_i)$. Let $i \in \N_0$. We show that there exists an edge from $(q_{1,i},q_{2,i})$ to $(q_{1,i+1},q_{2,i+1})$. We know that $(\#,\alpha_{i+1})$ is accepted from $q_{1,i}$, $(\alpha_{i+1},\beta_{i+1})$ leads from $q_{2,i}$ to $q_{1,i+1}$ and $(\beta_{i+1},\beta_{i+1})$ leads from $q_{2,i}$ to $q_{2,i+1}$. Therefore we can set $c := \alpha_{i+1}$ and $d := \beta_{i+1}$ in the definition of the edges of $G$ to prove the existence of the edge. Since there are infinitely many pairs $(q_{1,i},q_{2,i})$, we get an infinite path; therefore $G$ must have a cycle reachable from the pair $(q_{1,1},q_{2,1})$. Now observe that $q_{1,1} = \delta_T(q_0, (\alpha_0,\beta_0))$, $q_{2,1} = \delta_T(q_0, (\beta_0,\beta_0))$, $\alpha_0 \in \initial$ and $|\beta_0| = |\alpha_0|$.

For the converse, assume that $G$ has a cycle reachable from some pair $(q_{1,0},q_{2,0})$ with $\alpha_0$, $\beta_0 \in \Sigma^*$ such that $q_{1,0} = \delta_T(q_0, (\alpha_0,\beta_0))$, $q_{2,0} = \delta_T(q_0, (\beta_0,\beta_0))$, $\alpha_0 \in \initial$ and $|\beta_0| = |\alpha_0|$. For $i \in \N$, let $(q_{1,i},q_{2,i})$ be the $i$-th node reached by the infinite path consisting of the path from $(q_{1,0},q_{2,0})$ to the cycle and then going through the cycle forever. By the definition of $G$, there exist $(c_i)_{i \in \N_0}$, $(d_i)_{i \in \N_0}$ such that for all $i \in \N_0: |c_i| = |d_i| > 0$, $\delta_T(q_{1,i},(\#,c_i)) \in F_T$, $\delta_T(q_{2,i},(c_i,d_i)) = q_{1,i+1}$, and $\delta_T(q_{2,i},(d_i,d_i)) = q_{2,i+1}$. Let $\alpha_i := c_{i+1}$ and $\beta_i := d_{i+1}$ for $i \in \N$. The sequences $(\alpha_i)_{i \in \N_0}$ and $(\beta_i)_{i \in \N_0}$ now satisfy Lemma~\ref{alphabeta}.

This characterisation gives us a nondeterministic algorithm to check for the non-repeating case of abstract infinite reachability: One can construct $T$, guess a node $(q_{1,0},q_{2,0})$ of $G$, guess $\alpha_0 \in \initial$ and $\beta_0 \in \Sigma^*$ with $|\beta_0| = |\alpha_0|$ and check that $q_{1,0} = \delta_T(q_0, (\alpha_0,\beta_0))$ and $q_{2,0} = \delta_T(q_0, (\beta_0,\beta_0))$. Then, one can repeatedly guess the next node $(q_{1,i},q_{2,i})$ by guessing $\alpha_i$ and $\beta_i$ until one reaches the same node twice. This algorithm takes double-exponential space as $T$ can be of double-exponential size. However, we can do better by running on the transducer $\overline T$ which recognizes $\Relation{T} = \overline{\potreach} \cup (\Sigma^* \times \overline{\Language{A}})$.

$\overline T = (Q_{\overline T}, \Sigma, \delta_{\overline T}, Q_{0\overline T}, F_{\overline T})$ has exponential size, see Theorem~\ref{thm:abs}. Like in the power set construction for NFAs, sets of states of $\overline T$ correspond to states of $T$. We redefine the graph $G$ accordingly: the set of nodes becomes $(Q_{\overline T})^2$, the states of $T$ in the definition of the edges become sets of states of $\overline T$, the set $F_T$ is replaced by $2^{\overline {F_{\overline T}}}$. Correspondingly, the nondeterministic algorithm first guesses two sets $Q_{1,0}$, $Q_{2,0} \in 2^{Q_{\overline T}}$. All guesses of $\alpha_i$ and $\beta_i$ are done letter by letter to avoid storing words of double-exponential size. For example, the algorithm guesses $\alpha_0$ and $\beta_0$ letter by letter simultaneously and runs all checks on the fly. To check for a cycle in $G$, the algorithm nondeterministically memorizes a visited node $(Q_{1,i}, Q_{2,i})$ and terminates when it visits that node again. Since all nodes and automata have at most exponential size, this algorithm only needs exponential space. Apply now $\textsf{NEXPSPACE} = \EXPSPACE$.
\end{proof}

\section{Proof of Theorem~\ref{thm:abssuretermination}}
\label{app:abssuretermination}
\newcommand{\cons}{\varphi}

\abssuretermination*

We show that abstract sure termination, i.e.\ the version of \textsc{Abstract Liveness} where $L = \configurations$, is \EXPSPACE-hard by a reduction from the length-preserving version of \textsc{Abstract Liveness}.

Let $\system = (\initial, \trafun)$, $\interpretation$, $A$ be an instance of \textsc{Abstract Liveness}; let $\Sigma$, $\Gamma$ be the alphabets of the configurations and constraints, respectively. We construct an RTS $\system' = (\initial', \trafun')$ over an alphabet $\Sigma'$ and an interpretation $\interpretation'$ such that $\system'$ has an infinite potential run iff $\system$ has an infinite potential run visiting $L := \Language{A}$ infinitely often.

First, let $n$ be the number of states of $A$. We artificially change $A$ to a DFA $A'$ by adding numbers to letters in $\Sigma$, i.e.\ changing the alphabet of $A$ from $\Sigma$ to $\Sigma \times [n]$ and adding numbers to labels in $A$ such that no two transitions from the same state have the same label. By adding a trap state, $A'$ becomes a DFA. To avoid cluttered notation, we write $\Sigma$ for $\Sigma \times [n]$, denote words over $\Sigma$ as pairs $(c,x)$, write $c$ for $\{(c,x) \mid x \in [n]\}$, and denote the set of all words over $\Sigma$ by $\configurations$. Let $L' := \Language{A'}$. Note that $L'$ is not (necessarily) equal to $\{(c,x) \mid c \in L, x \in [n]^*, |x|=|c|\}$, only a subset.

Having done that, we define our RTS. Let $\Sigma_1 = \{a_1 \mid a \in \Sigma\}$ be a copy of $\Sigma$. For a configuration $c \in \configurations$, we denote the word which results from replacing every letter of $c$ with its copy from $\Sigma_1$ by $c_1$, and write $\configurations_1$ for $\{c_1 \mid c \in \configurations\}$. Let $\Sigma' := \Sigma_1 \cup \Sigma \cup (\Sigma \times \Sigma)$, $\configurations' := \Sigma'^*$, $\initial' := \initial$ and $\trafun' := \{(c_1,c'_1) \mid (c,c') \in \trafun\} \cup \{((c,c),(c,c)) \mid (c,c) \in \trafun\}$. Intuitively, the part of $\system'$ with configurations over $\Sigma_1$ behaves the same way as $\system$, and there are additional configurations over $\Sigma$ and over $\Sigma \times \Sigma$ which have no incoming or outgoing transitions except self-loops on configurations of the form $(c,c)$. Note that all runs of $\system'$ immediately terminate, but since \emph{abstract} sure termination is analysed, this is irrelevant.

We define $\interpretation'$ as the union of four deterministic transducers $\interpretation_1, \interpretation_2, \interpretation_3, \interpretation_4$. Firstly, we define $\interpretation_1$ and $\interpretation_2$ such that for all constraints $\cons \in \Gamma^*$ of $\system$, we have $\interpretation_1(\cons) = \{c_1 \mid c \in \interpretation(\cons)\} \cup (\interpretation(\cons) \times \configurations)$ and $\interpretation_2(\cons) = \{c_1 \mid c \in \interpretation(\cons)\} \cup (\configurations \times \interpretation(\cons))$.

$\interpretation_3$ adds, for every $(c,x) \in \configurations$, a constraint for the set $\{(c,x)\} \cup (L' \times \configurations) \cup (\configurations \times L')$. Since $L'$ is given by a DFA, one can construct a DFA for $(L' \times \configurations) \cup (\configurations \times L')$, and thus a DFA for $\interpretation_3$, in polynomial time.

Let $\Gamma_x = \{x_{ab} \mid a,b \in \Sigma \times [n]\}$. The DFA $\interpretation_4$ looks like this:
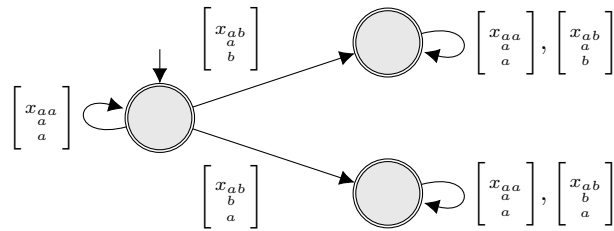
\begin{center}\begin{tikzpicture}[node distance=2.2cm,auto]
\tikzstyle{every state}=[fill={rgb:black,1;white,10}]

\node[state,initial above,initial text=,accepting] at (0,0) (q0) {};
\node[state,accepting] at (3,1) (q1) {};
\node[state,accepting] at (3,-1) (q2) {};

\path[->]
(q0) edge [loop left] node {$\bigg[{x_{aa} \atop {a \atop a}}\bigg]$} ()
edge node {$\bigg[{x_{ab} \atop {a \atop b}}\bigg]$} (q1)
edge [swap] node {$\bigg[{x_{ab} \atop {b \atop a}}\bigg]$} (q2)
(q1) edge [loop right] node {$\bigg[{x_{aa} \atop {a \atop a}}\bigg], \bigg[{x_{ab} \atop {a \atop b}}\bigg]$} ()
(q2) edge [loop right] node {$\bigg[{x_{aa} \atop {a \atop a}}\bigg], \bigg[{x_{ab} \atop {b \atop a}}\bigg]$} ();
\end{tikzpicture}\end{center}
A transition of the form $\bigg[{x_{ab} \atop {a \atop b}}\bigg]$ stands for all transitions of the form $\bigg[{x_{ab} \atop {a \atop b}}\bigg]$ where $a \neq b$. Missing transitions lead to an implicit trap state. Since $|\Gamma_x|$ is quadratic in $|\Sigma \times [n]|$, $\interpretation_4$ is polynomial in the input.

Intuitively, the letter $x_{ab}$ means that at the corresponding position, one of the configurations has an $a$ and the other has a $b$. For $(\varphi,(c,c'))$ to be accepted in the upper state of $\interpretation_4$, $c$ (resp.\ $c'$) must have exactly the letters from the first (resp.\ second) position of the correspoding letters in $\Gamma_x$, and vice versa for acceptance in the bottom state. Hence $\interpretation_4$ generates, for every $c, c' \in \configurations$, a constraint for the set $\{(c,c'), (c',c)\}$.

We can now describe the potential reachability relation $\potreach'$ of $\system'$. First, note that all constraints generated by $\interpretation_3$ are inductive. It follows that for all distinct $c,c' \in \configurations'$, we have $(c,c') \notin \potreach'$ if $c' \in \configurations$, or if $c' \in \configurations_1$ and $c \notin \configurations_1$. Now let $c,c',c'' \in \configurations$. Because of $\interpretation_3$, we have $(c, (c',c'')) \notin \potreach'$ if $c' \notin L'$ and $c'' \notin L'$. If $c' \in L'$ or $c'' \in L'$, then $(c, (c',c'')) \in \potreach'$ is equivalent to $(c,c'), (c,c'') \in \potreach$ (enforced by $\interpretation_1$ and $\interpretation_2$). Next, all sets of the form $\{(c,c'), (c',c)\}$ are also inductive and have constraints generated by $\interpretation_4$. Hence, for every $c,c',d,d' \in \configurations$, we can have $((c,c'),(d,d')) \in \potreach$ only if $\{c,c'\} = \{d,d'\}$.

Now assume that $\system$ has a potential run visiting $L$ infinitely often. By Lemma~\ref{lem:aandb}(a) (recall that $\system$ is length-preserving), there exist $c_0 \in \initial$, $c \in L$ such that $(c_0,c) \in \potreach$ and $(c,c) \in \trafun \cup (\potreach \setminus \Id)^2$. In particular, there exists a number $x$ such that $(c,x) \in L'$, and thus $(c_0, (c,x)) \in \potreach'$. (In the following, we just write $c$ for $(c,x)$.) If $(c,c) \in \trafun$, then $((c,c),(c,c)) \in \trafun'$; if $(c,c) \in (\potreach \setminus \Id)^2$, then there exists $c' \neq c$ such that $(c,c'), (c',c) \in \potreach$, so we have $((c,c'),(c',c)),((c',c),(c,c')) \in \potreach'$, i.e.\ $((c,c'),(c,c')) \in (\potreach' \setminus \Id)^2$. This shows that Lemma~\ref{lem:aandb}(a) holds for $\system'$, proving the existence of an infinite potential run.

For the converse, assume that $\system'$ has an infinite potential run. By Lemma~\ref{lem:aandb}(a), there exist configurations $(c_0,x) \in \initial'$ and $\mathbf{c} \in \configurations'$ such that $((c_0,x),\mathbf{c}) \in \potreach'$ and $(\mathbf{c},\mathbf{c}) \in \trafun' \cup (\potreach' \setminus \Id)^2$. We have already shown that because of $\interpretation_3$, $((c_0,x),\mathbf{c}) \in \potreach'$ implies that $\mathbf{c} = ((c,x),(c',x))$ with $(c,x) \in L'$ or $(c',x) \in L'$, and additionally, $(c_0,c), (c_0,c') \in \potreach$. Let w.l.o.g.\ $(c,x) \in L'$. Then $c \in L$. If $(\mathbf{c},\mathbf{c}) \in \trafun'$, then we have $c' = c$ and $(c,c) \in \trafun$ by definition of $\trafun'$, and thus Lemma~\ref{lem:aandb}(a) holds for $\system$. If $(\mathbf{c},\mathbf{c}) \in (\potreach' \setminus \Id)^2$, then there exists $\mathbf{c'} \in \configurations'$ such that $\mathbf{c'} \neq \mathbf{c}$ and $(\mathbf{c},\mathbf{c'}), (\mathbf{c'},\mathbf{c}) \in \potreach'$. We have already shown that this implies $\mathbf{c'} = ((d,x),(d',x))$ with $\{c,c'\} = \{d,d'\}$. Since $\mathbf{c'} \neq \mathbf{c}$, we must have $c \neq c'$ and $\mathbf{c'} = ((c',x),(c,x))$. Now the constraints generated by $\interpretation_1$ and $\interpretation_2$ imply that $(c,c'), (c',c) \in \potreach$. Hence $(c,c) \in (\potreach \setminus \Id)^2$, and Lemma~\ref{lem:aandb}(a) holds for $\system$, i.e.\ $\system$ has a run visiting $L$ infinitely often.

\end{document}

%% file: notation.tex
\usepackage{amsmath}
\usepackage{complexity}
\usepackage{tikz}
\usepackage[vlined,linesnumbered]{algorithm2e}
\usepackage{csquotes}
\usepackage{nicefrac}
\usepackage{array}
\usepackage{ragged2e}
\usepackage{bm}
\newcolumntype{R}[1]{>{\RaggedLeft\arraybackslash}p{#1}}
\usetikzlibrary{arrows.meta,automata,positioning,petri}

\tikzset{>={Latex[length=2mm,width=2mm]}}

\newcommand{\system}{\mathcal{S}}
\newcommand{\initial}{\mathcal{I}}

\newcommand{\interpretation}{\mathcal{V}}
\newcommand{\configurations}{\mathcal{C}}

\newcommand{\trafun}{\Delta}
\newcommand{\reach}{\textit{Reach}}
\newcommand{\potreach}{\textit{PReach}}

\newcommand{\N}{\mathbb{N}}

\newcommand{\Language}[1]{\lang{L}({#1})}
\newcommand{\Relation}[1]{\lang{R}({#1})}

\newcommand{\vtuple}[2]{\big[{{#1} \atop {#2}}\big]}

\definecolor{nicebg}{HTML}{f6f0e4}
\definecolor{nicered}{HTML}{7f0a13}
\definecolor{nicebgred}{HTML}{f2e7e8}
\definecolor{niceblue}{HTML}{104354}
\definecolor{nicebgblue}{HTML}{e8edee}
\definecolor{nicegreen}{HTML}{217516}
\definecolor{nicebggreen}{HTML}{e9f1e8}
\definecolor{nicepurple}{HTML}{884bab}
\definecolor{nicebgpurple}{HTML}{f3edf7}
\definecolor{niceorange}{HTML}{d27c11}
\definecolor{nicebgorange}{HTML}{fbf2e8}
\definecolor{nicepink}{HTML}{e95f9f}
\definecolor{nicebgpink}{HTML}{fdeff6}
\definecolor{niceredlight}{HTML}{c9888d}
\definecolor{nicebluelight}{HTML}{78a4b8}
\definecolor{nicegreenlight}{HTML}{76de68}
\definecolor{nicepurplelight}{HTML}{bc87db}
\definecolor{niceredbright}{HTML}{bd0310}
\definecolor{nicebgredbright}{HTML}{f9e6e8}
\definecolor{nicebluebright}{HTML}{197b9b}
\definecolor{nicebgbluebright}{HTML}{e8f2f5}

\makeatletter
\newcommand\xleftrightarrow[2][]{\ext@arrow 0099{\longleftrightarrowfill@}{#1}{#2}}
\def\longleftrightarrowfill@{\arrowfill@\leftarrow\relbar\rightarrow}
\makeatother